\documentclass[conference]{IEEEtran}
\usepackage{cite}
\usepackage{amsmath,amssymb,amsfonts}
\usepackage[linesnumbered,lined,ruled,commentsnumbered]{algorithm2e}
\usepackage{graphicx}
\usepackage{textcomp}
\usepackage{xcolor}
\usepackage{hyperref}
\usepackage{amsthm}
\usepackage[noabbrev,capitalize]{cleveref}
\usepackage{graphicx}
\usepackage{caption}
\usepackage{subcaption}
\usepackage[colorinlistoftodos]{todonotes}
\usepackage{xspace}
\usepackage{amsthm}
\usepackage{paralist}
\usepackage{colortbl}
\usepackage{array,booktabs}
\usepackage{multirow}

\graphicspath{{./fig/}}

\newcommand{\pp}[1]{\medskip\noindent\textbf{\textit{#1}.}\xspace}

\newcommand{\ema}[1]{\ensuremath{#1}}

\newcommand{\rev}[1]{\color{black}#1\color{black}\xspace}

\newtheorem{problem}{Problem}
\newtheorem{theorem}{Theorem}

\newtheorem{lemma}{Lemma}

\newtheorem{remark}{Remark}

\newtheorem*{assumption*}{\assumptionnumber}
\providecommand{\assumptionnumber}{}
\makeatletter
\newenvironment{assumption}[1]
 {%
  \renewcommand{\assumptionnumber}{Assumption #1}%
  \begin{assumption*}%
  \protected@edef\@currentlabel{#1}%
 }
 {%
  \end{assumption*}
 }
\makeatother

\newcommand{\cosched}{\textsc{Co-Sched}\xspace}
\newcommand{\coalloc}{\textsc{Co-Alloc}\xspace}
\newcommand{\evalloc}{\textsc{Ev-Alloc}\xspace}
\newcommand{\feasible}{\textsc{Feasible}\xspace}
\newcommand{\order}{\textsc{Order}\xspace}

\newcommand{\ideal}{\ema{\textsc{Ideal}}\xspace}
\newcommand{\transit}{\ema{\textsc{In-transit}}\xspace}
\newcommand{\increasing}[1]{\ema{\textsc{Increasing-#1\%}}\xspace}
\newcommand{\decreasing}[1]{\ema{\textsc{Decreasing-#1\%}}\xspace}

\newcommand{\insitu}{in~situ\xspace}

\newcommand{\nsteps}{\ema{n_{steps}}\xspace}

\SetKwComment{Comment}{/* }{ */}
\SetKwInOut{Input}{Input}
\SetKwInOut{Output}{Output}
\SetKwProg{Fn}{Function}{:}{}
\SetKwFunction{Fcoalloc}{\coalloc}
\SetKwFunction{Ffeasible}{\feasible}
\SetKwFunction{Forder}{\order}

\newcommand{\naa}{\ema{\widetilde{N}}\xspace}

\begin{document}

\title{Co-scheduling Ensembles of In Situ Workflows}

\author{
    \IEEEauthorblockN{
        Tu Mai Anh Do\IEEEauthorrefmark{1},
        Lo\"ic Pottier\IEEEauthorrefmark{1},
        Rafael Ferreira da Silva\IEEEauthorrefmark{3},
        Frédéric Suter\IEEEauthorrefmark{3},
        Silvina~Ca\'ino-Lores\IEEEauthorrefmark{4},\\
        Michela Taufer\IEEEauthorrefmark{4},
        Ewa Deelman\IEEEauthorrefmark{1}
    }
    \IEEEauthorblockA{
        \IEEEauthorrefmark{1}Information Sciences Institute, University of Southern California, Marina Del Rey, CA, USA \\
        \IEEEauthorrefmark{3}Oak Ridge National Laboratory, Oak Ridge, TN, USA \\
        \IEEEauthorrefmark{4}University of Tennessee at Knoxville, Knoxville, TN, USA \\
        \{tudo,lpottier,deelman\}@isi.edu, \{silvarf,suterf\}@ornl.gov, 
        \{mtaufer,scainolo\}@utk.edu
    }
}


\maketitle

\begin{abstract}
Molecular dynamics (MD) simulations are widely used to study large-scale 
molecular systems. HPC systems are ideal platforms to run these studies, 
however, reaching the necessary simulation timescale to detect rare processes 
is challenging, even with modern supercomputers. To overcome the timescale 
limitation, the simulation of a long MD trajectory is replaced by multiple 
short-range simulations that are executed simultaneously in an ensemble of 
simulations. 
Analyses are usually 
co-scheduled with these simulations to efficiently process large volumes of data
generated by the simulations at runtime, thanks to in situ techniques. 
Executing a workflow ensemble of simulations
and their in situ analyses requires efficient co-scheduling strategies and
sophisticated management of computational resources so that they are not
slowing down each other.
In this paper, we propose an efficient method to co-schedule simulations and in situ
analyses such that the makespan of the workflow ensemble is minimized.
We present a novel approach to allocate resources for a workflow ensemble under resource constraints by using a 
theoretical framework modeling the workflow ensemble's execution. 
We evaluate the proposed approach using an accurate simulator based on the WRENCH 
simulation framework 
on various 
workflow ensemble configurations. 
Results demonstrate the significance of co-scheduling simulations 
and in situ analyses that couple data together to benefit from data locality, in which inefficient 
scheduling decisions can lead up to a factor 30 slowdown in makespan. 
\end{abstract}

\begin{IEEEkeywords}
workflow ensemble, in situ, co-scheduling, molecular dynamics, high-performance computing
\end{IEEEkeywords}

\section{Introduction}
\label{sec:intro} 

Molecular dynamics (MD) is one of the scientific simulations that benefit most from ever-increasing computational capabilities offered by modern high-performance computing (HPC) platforms. 
MD simulations serve as a productive
method to observe relevant processes of a molecular system at atomic
resolution. MD simulations discover the motion of the system through computing
atomic positions that evolve over time.
\rev{MD simulations regularly generate large amount of data that needs to be stored efficiently to disks.}
However, I/O capabilities have not evolved as fast as compute capabilities in contemporary HPC infrastructures, which causes an important performance bottleneck at scale~\cite{ian2017} due to large amount of data generated by the simulations that cannot be stored efficiently to disks. 
Consequently, the community has developed different approaches to overcome the I/O bottleneck---in this work, we focus on a particular one called \insitu. 
The \insitu paradigm shifts the data processing phase from traditional post-processing to iterative processing following a producer-consumer approach. Basically, instead of storing data produced by simulations on disks, simulations send data periodically using faster memories (e.g., DRAM, SSD) to \insitu analysis jobs that run concurrently to the simulations. Performing analyses \insitu does not only help to obtain on-the-fly insights into the molecular system at runtime, but it also reduces the burden on the file system.

Interesting molecular events require observing at a long enough simulation timescale, which is
compute-intensive for simulating large-scale molecular systems.
To overcome this timescale problem, a family of enhanced-sampling methods
utilizes multiple short-range simulations running within \emph{ensembles}
to efficiently explore different regions of the conformational space.
Ensemble-based computational methods~\cite{paolo2006ercfe,jeffrey2014mrsfe,okamoto2004geaes,riccardo2012sges} are gaining popularity in many
scientific domains using computational simulations, where computations are executed as jobs into a \emph{workflow ensemble}. 
A workflow ensemble~\cite{malawski2015ensemble} is usually composed of a number of inter-related 
workflows, for example identical workflows of simulation-analysis pipeline with different input parameters 
exploring the solutions space.
We consider the 
execution of ensembles of \insitu workflows in which concurrent jobs, i.e. simulations and analyses, are 
tightly coupled using \insitu processing, thereby usually co-scheduled together to efficiently process data at runtime.
However, sustaining complex co-scheduling for massive workflow ensembles of concurrent
applications will require sophisticated orchestration tailored 
for the management of workflow ensembles~\cite{ahn2020flux,peterson2019enabling}. 
 
In this paper, we 
aim to guide the above 
emerging orchestration runtime. The key challenge when orchestrating the
workflow ensembles at scale is to efficiently schedule the simulations and \insitu 
analyses as they continuously exchange data. 
The intersection of workflow ensembles and \insitu is particularly challenging as 
\insitu itself has intricate, complex communication patterns, which is further 
convoluted by the number of concurrent executions at scale.
%
Intuitively, co-scheduling the analyses with their associated simulation on the same computing resources improves data locality, but required resources may not be 
consistently available.
Thus, upon resource constraint, how do we co-schedule workflow ensembles on 
a HPC machine?
A key distinguishing factor is to efficiently generate large ensembles to optimize scientific discoveries. 
Another major challenge is to determine resource requirements such 
that the resulting performance of executing them together in the workflow 
ensemble is maximized. Therefore, we design 
efficient co-scheduling strategies and resource assignments
for the workflow ensemble of simulations and \insitu analyses. 
Our contributions are as follows:
\begin{itemize}
    \item We develop a mathematical model that characterizes iterative execution and data coupling behavior of simulations and \insitu analyses in a workflow ensemble.
    \item We use this model to derive optimal co-scheduling strategies and compute resource allocations for the simulations and \insitu analyses in the workflow ensemble under constraints of the available computing resources.
    \item We evaluate our proposed scheduling solutions using an accurate simulator by studying different co-scheduling scenarios on various configurations of the workflow ensembles. 
\end{itemize}

\section{Related Work}
\label{sec:related}

Designing scheduling algorithms for \insitu workflows such that available resources are efficiently utilized is challenging as optimizing individual workflow components does not ensure \rev{that} the end-to-end performance of the workflow is optimal.
The \rev{larger the} number of components the workflow has, the larger the multi-parametric space needed to explore.
Specifically, workflow ensembles consist of many simulations and \insitu analyses running concurrently, thus the combination of their data coupling behaviors significantly increases the complexity of the scheduling decision\rev{s}. 
Few efforts have been attempted to solve this problem on a single \insitu workflow.
Malakar~{\it et al.} formulated the \insitu coupling between the simulation and the analysis as a mixed-integer linear programming problem~\cite{malakar2016oec} and derived the frequency to schedule the analysis executed \insitu on different set of nodes.
Aupy~{\it et al.} proposed a greedy algorithm~\cite{aupy2019mha} to schedule a set of \insitu analyses on available resources with memory constraints.
Taufer~{\it et al.} introduced a 2-step model~\cite{taufer2019cis} to predict the frequency of frames to be analyzed \insitu such that computing resources are not underutilized while the simulation is not slowed down by the analysis.
Our proposed approach for scheduling is not restricted to a single \insitu coupling. 

Co-scheduling strategies have been recently proposed 
to deliver better resource utilization and increasing overall application throughput~\cite{breibart2017dcd,zacarias2019icw,kuchumov2021ane}. 
However, only few works incorporated co-scheduling into the \insitu workflows, where components (simulations and analyses) in the workflow are tightly coupled together.
Sewell~{\it et al.} proposed a thorough performance study~\cite{sewell2015lca} for co-scheduling \insitu components that leverages memory system for data staging with low latencies.
Aupy~{\it et al.} introduced an optimization-based approach~\cite{aupy2018chw} to determine the optimal fraction of caches and cores allocated to \insitu tasks.
Due to the lack of capability to submit a batch of concurrent applications with conventional scheduler, co-scheduling was not fully-supported for the workflow ensembles at production scale.
To overcome this limitation, Flux~\cite{ahn2020flux} was designed as a hierarchical resource manager to meet the needs of co-scheduling large-scale ensemble of various simulations and \insitu analyses.
To the best of our knowledge, this paper is the first effort that addresses co-scheduling problem for \insitu workflows at ensemble level. 
\section{Model}
\label{sec:model}

The goal of this paper is to study the scheduling of complex ensembles of \insitu workflows~\cite{do2022pae} on parallel machines, where every workflow within an ensemble has several concurrent jobs coupled together using \insitu processing, either be a \emph{simulation} or an \emph{analysis}.
Unfortunately, one might not have enough compute nodes to run concurrently all simulations and analyses and, will have to make scheduling choices: (\emph{i}) which simulation and analysis should run together on the same resources and (\emph{ii}) how much resources should be allocated to each.

\pp{Our approach}
The core of this work is to propose a powerful but versatile model to express complex data dependencies between simulations and \insitu analyses in a workflow ensemble, and based on that design appropriate scheduling strategies. The chosen approach is to model \rev{the scheduling problem under simplified conditions  where we make several assumptions to reduce} the large space of exploration in an ensemble of \insitu workflows. Then, we demonstrate \rev{that} solutions designed for the simplified world, also perform well \rev{in} using realistic settings.

\subsection{Couplings}
A simulation $S_i$ represents a job in an ensemble that simulates a state of interest in the molecular system.
An analysis, denoted as $A_j$, represents a computational kernel that is suitable to analyze \insitu data produced by a simulation, e.g. transforming high-dimensional molecular structures to eigenvalue metadata~\cite{johnston2017isda}.
Regularly, the output resulted by \insitu analyses are later post-processed, e.g. synthesized to infer conformational changes of the molecular system in a post-hoc analysis.
However, in the scope of this paper, we focus on the part of the simulation-analysis pipeline where they are executed simultaneously in an ensemble as it is more challenging to manage. 
We denote the set of simulations and analyses in an ensemble as $S$ and $A$ respectively.
In this work, each workflow of an ensemble is composed of one simulation and at least one analysis.
The analyses are executed in parallel with the simulation and periodically processed simulation data generated after each simulation steps in an iterative fashion. This essence of \insitu processing expresses the important notion of \emph{coupling}.

Coupling indicates \rev{which simulation and analyses have} a data dependency, i.e. data produced by the simulation is analyzed by the corresponding analyses.
Scientists usually decides beforehand which analyses have to process the data from a given simulation, we call this mapping a \emph{coupling} (see~\cref{fig:mapping}). Therefore, couplings are decided beforehand as an input of the problem.
More formally, 
a coupling $p : S \rightarrow A$ defines which analyses in $A$ are coupled with which simulation in $S$, then for every $S_i \in S$, let $p(S_i)$ be the set of analyses that couple data with the simulation $S_i$. 
For example, in~\cref{fig:mapping}, we have $p(S_1) = \{A_1, A_3\}$.
Now that we have defined the notion of couplings, we need to define different notions around \emph{co-scheduling}.

%

%
%


\subsection{Co-scheduling}

The notion of co-scheduling is at the core of this work, \emph{co-scheduling} is defined as running multiple applications concurrently on the same set of compute nodes, each application using a fraction of the number of cores per node.
We assume the total number of cores of each application is distributed evenly among the compute nodes the application is co-scheduled on, which is similar to the manner of distributing resources in emerging schedulers, e.g. Flux~\cite{ahn2020flux}, for co-scheduling concurrent jobs of a massive ensemble in next-generation computers.
In this paper, we distinguish two scenarios in which we map the simulations and analyses to computational resources:
\begin{compactitem}
	\item \emph{Co-scheduling}: $S_i$ and $A_j$ run on the same compute nodes and thus have access to the shared memory;
	\item \emph{In transit}: $S_i$ and $A_j$ run on two different set of nodes and sending data from $S_i$ to $A_j$ involves network communications.
\end{compactitem}
\rev{In an attempt to avoid performance degradation in co-scheduling, in this work, we consider that simulations, which are compute-intensive, are not co-scheduled together~\cite{dauwe2014mep}}. 
Now we need to define several other important terms related to co-scheduling.

\begin{figure}[!ht]
	\centering
	\includegraphics[width=\linewidth]{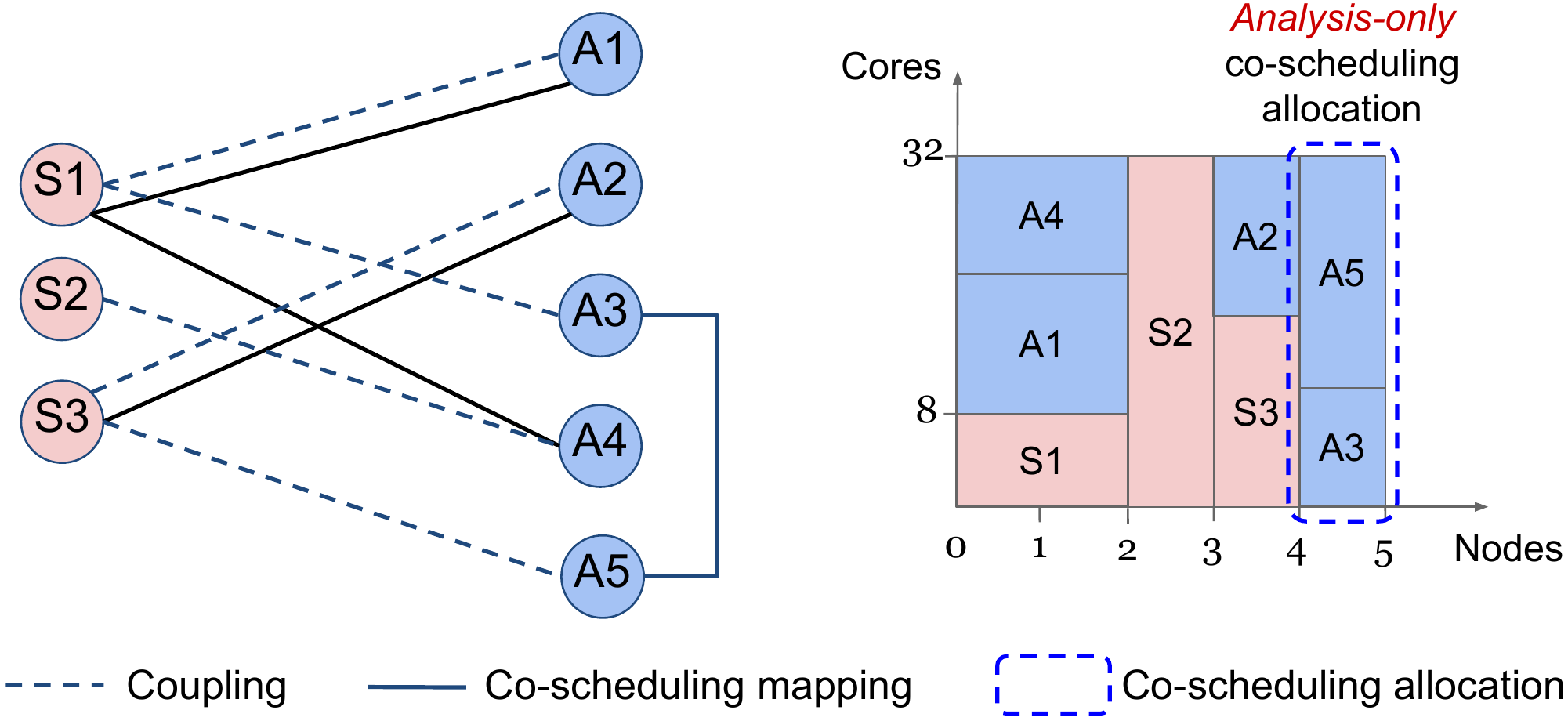}
	\caption{Illustration of co-scheduling and coupling notions.}
	\label{fig:mapping}
\end{figure}

\pp{Co-scheduling mapping} 
A \emph{co-scheduling mapping} $m$ defines how $S$ and $A$ are co-scheduled together.
For a given mapping $m$, we denote by 
$m(S_i)$ the set of analyses co-scheduled with $S_{i}$. 
For example, with the mapping illustrated in~\cref{fig:mapping}, $m(S_1) = \{ A_1, A_3\}$ and $m(S_3) = \{ A_2\}$.
\begin{assumption}{1}
In this model, we ignore co-scheduling interference, such as cache sharing, that could degrade the performance when multiple applications share resources~\cite{dauwe2014mep}.
\end{assumption}

\pp{Co-scheduling allocation} 
A co-scheduling allocation represents a set of applications which are co-scheduled together (i.e., share computing resources). 
Let $\langle \{S_i, A_j\} \rangle$ denote a co-scheduling allocation where simulation $S_i$ is co-scheduled with analysis $A_j$. Recall that, we have at most one simulation per allocation, but we can have multiple analyses.
Formally, for a given co-scheduling mapping $m$, a simulation $S_i$ and analyses $m(S_i)$ that are co-scheduled with $S_i$, the corresponding co-scheduling allocation is denoted by $\langle S_i \cup m(S_i) \rangle$.
However, not all analyses are required to be co-scheduled with a simulation, there exists co-scheduling allocations in which only analyses are co-scheduled together, called \emph{analysis-only co-scheduling allocations}.
For example, in~\cref{fig:mapping}, $\langle \{A_3, A_5\} \rangle$ is an analysis-only co-scheduling allocation.
A co-scheduling allocation also determines the amount of computer resources assigned for each co-scheduled application. 
Let $n_{x}$ denote the number of nodes assigned to application $x$. 
For instance, if $S_i$ and $A_j$ are co-scheduled on the same co-scheduling allocation, 
then $n_{S_i} = n_{A_j}$.
Let $c_{x}$ be the number of cores per node assigned to $x$.
As each co-scheduled application takes a portion of cores in the total of $C$ cores a node has, then for every co-scheduling allocation $\langle X \rangle$, we have
%
	$\sum_{x \in X} c_{x} \leq C$ .
%
In~\cref{fig:mapping}, since $A_1$ and $A_4$ are co-scheduled with $S_1$ on a co-scheduling allocation spreading in 2 compute nodes, then $n_{S_1} = n_{A_1} = n_{A_4} = 2$.
And since $c_{S_1} = 8$, then $S_1$ occupies 8 cores each node, so 16 cores in total.


\subsection{Execution platform}
We consider a traditional parallel computing platform with $N$ identical compute nodes, each node has $C$ identical cores, a shared memory of size $M$ and a bandwidth of $B$ per node. 
In addition, we consider that the interconnect network between nodes is a fully-connected topology, hence the communication time is the same for any pairs of nodes.

\subsection{Application model}

Because of the \insitu execution, simulations and analyses are
executed iteratively, in which iterations exhibit consistent
behavior~\cite{do2021lightweight}.
Specifically, the simulation executes certain
compute-intensive tasks each iteration and then generates output data. The analysis,
in every iteration, subscribes the data produced by the simulation and
performs several specific \rev{computations}. Since the complexity of computation
and the amount of data are identical across iterations, the time to
execute one iteration is approximately the same across 
iterations. Based on this reasoning 
and, by assuming the simulations and analyses have the same number 
of iterations ($\nsteps$), we simply consider the execution time of a single 
iteration as the entire execution time can be derived from 
multiplying the time taken by a single iteration by $\nsteps$~\cite{do2021lightweight}.

\pp{Computation model}
Let us now denote by $t(x)$ the time to execute one iteration of the iterative application $x$ ($x$ can either be a simulation or an analysis). 
For readability, we pay particular attention to the class of perfectly parallel applications, even though the approach can be generalized to other models, such as Amdahl's Law~\cite{amdahl1967vspa}.  
\begin{assumption}{2}
    Every simulation and analysis follow a perfectly parallel speedup model~\cite{herlihy2012}.
\end{assumption}
In other words, a job $x$ that runs in time $t_x(1)$ on one core will run in time $t_x(1) / c $ when using
$c$ cores.
Thus, we can express the time to execute one iteration of the simulation as follows
\begin{equation}
    \label{eq:ts}
    t(S_{i}) = \frac{t_{S_{i}}(1)}{n_{S_{i}} c_{S_{i}}},
\end{equation}
where $S_{i}$ is running on $n_{S_{i}}$ compute nodes with $c_{S_{i}}$ cores each node. 

\pp{Communication model}
%
In \insitu processing, data is usually stored on node-local storage~\cite{do2018enabling} or local memory~\cite{kim2015pnm} to reduce the overhead associated with data movement. 
In this work, we process data residing in memory to allow near real-time data analysis. 
Data communications are modeled as
follows: (\emph{i}) $S_i$ writes its results into local
memory of the nodes it runs on, 
and (\emph{ii}) each analysis coupled with $S_i$ reads from that
memory the data it needs whether locally or remotely depends on $A_j$ is co-scheduled with $S_i$ or not.

%
\begin{assumption}{3}
\label{assum:comm}
Intra-node communications (i.e., shared memory) are considered negligible compared to inter-node communications (i.e., distributed memory).
\end{assumption}
In this model, we consider that writing/reading to/from the local memory of a node (e.g., the RAM) is negligible, 
therefore the cost of communicating data between simulation and analysis co-scheduled together is approximately close to zero. However, communications going over the
interconnect have to be modeled. 
For example, in~\cref{fig:mapping}, $A_3$, $A_4$ and $A_5$ are not co-scheduled with the simulations they couple with, hence they incur overhead of remote transfers.
To model these 
inter-node communications, we adopt a classic linear bandwidth 
model~\cite{williams2009} where it takes $\frac{S}{B}$ to 
communicate a message of size $S$, where $B$ is the bandwidth. 
%
Let $V_{A_{j}}$ be the amount of data received and processed by $A_j$.
Remember that $A_j$ runs on $n_{A_j}$ compute nodes with $c_{A_j}$ core per node.
Then, based on the linear bandwidth model, the I/O time of $A_j$ can be expressed as 
%
     $V_{A_{j}} / B \, n_{A_{j}}$,
where $B$ is the maximum bandwidth of one compute node.
Even though the bandwidth is actually variable as it is shared by multiple I/O operations during the execution, we use the maximum bandwidth of a compute node for an optimistic consideration.

\pp{Execution time}
Finally, we combine
computation and communication model to get actual time to execute one iteration:
\begin{equation}
    \label{eq:ta}
    t(A_{j}) = 
    \begin{cases}
        \dfrac{t_{A_{j}}(1)}{n_{A_{j}} c_{A_{j}}} & \text{if $A_j$ is co-scheduled with }\\[-10pt]
        & \text{the simulation it couples with;} \\
        \dfrac{t_{A_{j}}(1)}{n_{A_{j}} c_{A_{j}}} + \dfrac{V_{A_{j}}}{B \, n_{A_{j}}}& \text{otherwise;}
    \end{cases}       
\end{equation}
while $t(S_i)$ is simply~\cref{eq:ts} as communications to local memory are free (cf. Assumption~\ref{assum:comm}).
~\cref{table:notations} summarizes the notation introduced in this section.

\begin{table}[!ht]
    \centering
    {
        \setlength{\tabcolsep}{3pt}
        \begin{tabular}{ll}
            \toprule
            Notation & Description \\
            \midrule
            $S_i, A_j$ & Set of simulations and analyses ($1 \leq i,j \leq n$)  \\
            $N, C$ & Total number of nodes and number of cores per node\\
            $B$ & Maximum bandwidth per node \\
            $p(S_i)$ & Set of analyses that couple data with simulation $S_i$ \\
            $V_{A_j}$ & Amount of data received by analysis $A_j$ \\
            $m(S_i)$ & Set of analyses co-scheduled with simulation $S_i$ \\
            $\langle \{S_i, A_j\} \rangle$ & Co-scheduling allocation where $S_i$ and $A_j$ are co-scheduled\\
            $n_x, c_x$ & Number of nodes and cores/node assigned to $x$ \\
            $t_x(1)$ & Time to execute one iteration of $x$ on one core \\
            $t(x)$ & Actual time to execute one iteration of $x$ on $n_x, c_x$\\
            $P^{NC}$ & Set of analyses co-scheduled on analysis-only allocations \\
            \bottomrule
        \end{tabular}
    }
    \caption{Notations.}
    \label{table:notations}
\end{table}

\subsection{Makespan}
\label{subsec:makespan}
We have defined the execution time for the simulation and the analysis and, armed with the notion of \emph{co-scheduling mapping} previously discussed, we are ready to define the makespan of an ensemble of \insitu workflows. 
Let $\nsteps$ denote the number of iterations, recall that we assume the simulations and analyses in an ensemble have the same number of iterations. 
The makespan can now be expressed as follows:
\begin{align}
    \label{eq:makespan}
    \textsc{Makespan} &= 
    \max_{S_i \in S, A_j \in A} \; (t(S_i), t(A_j)) \times \nsteps 
\end{align}


We are now ready to define different scheduling problems tackled in this work such that this makespan is minimized.

\section{Scheduling Problems}
\label{sec:coscheduling}

We consider that simulations and analysis are moldable jobs~\cite{feitelson1997} (i.e., \rev{resources allocated to a job are determined before its execution starts}). In addition, all the results presented in this section are using rational amount of resources, i.e. number of nodes and cores are rational. Our approach is to design optimal, but rational solutions and then adapt these results to a more realistic setup where resources are integer, by using rounding heuristics discussed in~\cref{sec:allocation}.
Rounding heuristic is a well-known practice commonly used in resource scheduling~\cite{stillwell2010raa,aupy2019mha} to round to integer values the rational values assigned to 
integer variables, which are number of nodes and cores in our case.

\subsection{Problems}
We consider two problems in this research, the first problem, \cosched, consists into finding a mapping that minimized the makespan of an ensemble of workflows given that we have a resource allocation, so for each simulation and analysis in the workflow we know on how many nodes/cores they are running. The second problem, \coalloc, is the reverse problem, from a given mapping we aim to compute an optimal resource allocation scheme.
\begin{problem}[\cosched]
Given a number nodes and cores assigned to simulations and analysis, find a co-scheduling mapping $m^{*}$ minimizing the makespan of the entire ensemble.
\end{problem}
%
%
\begin{problem}[\coalloc]
Given a co-scheduling mapping $m$, compute the amount of resources $n^{*}_{S_i},c^{*}_{S_i}$ assigned for each simulation $S_i \in S$ and $n^{*}_{A_j}, c^{*}_{A_j}$ for each analysis $A_j \in A$ such that the makespan of the entire ensemble is minimized.
\end{problem}

\pp{Ideal scenario}
To find a co-scheduling mapping which minimizes the makespan for \cosched, we start with
a \emph{ideal} co-scheduling mapping which is defined as a mapping where all analyses are co-scheduled with their coupled simulation. This scenario is denoted as \emph{ideal} as it implies that you have enough compute nodes to accommodate such mapping.
We show that the makespan is minimized under an ideal co-scheduling mapping by verifying the correctness of the following theorem:
\begin{theorem}
\label{tr:gm}
The makespan is minimized if, and only if, each analysis is co-scheduled with its coupled simulation.
\end{theorem}
\begin{proof}
The proof is given in~\cref{apdx:gm}.
\end{proof}
Based on~\cref{tr:gm}, we have a solution for~\cosched but under the strict assumption that we have enough resources to achieve such ideal mapping (i.e., co-schedule all analyses with their respective coupled simulations).
However, due to resource constraints, for instance, there is insufficient memory to accommodate all the analyses and their coupled simulation, or bandwidth congestion as multiple analyses performs I/O at the same time, that ideal co-scheduling mapping might be unachievable, and some analyses are not able to co-schedule with their coupled simulation.


\pp{Constrained resources scenario}
We therefore have to consider co-scheduling mappings in which there are analyses that are not co-scheduled with their coupled simulation.
There are numerous such co-scheduling mappings to explore, indeed each analysis that are not co-scheduled with its coupled simulation could potentially be co-scheduled with any other applications.
To reduce number of considered co-scheduling mappings, we show that the makespan is minimized if, and only if, analyses that are not co-scheduled with their coupled simulations are co-scheduled inside analysis-only co-scheduling allocations, i.e. without the presence of any simulation.
Based on this observation, we build a partition of analyses which are co-scheduled with their respective simulations and co-scheduled with other analyses (see~\cref{fig:mapping}).


\begin{theorem}
\label{tr:da}
Given a set of analyses that are not co-scheduled with their coupled simulation, the makespan of the ensemble is minimized when these analyses are co-scheduled within analysis-only co-scheduling allocations.
\end{theorem}
%
\begin{proof}
The proof is given in~\cref{apdx:da}.
\end{proof}
%
Note that, since analyses co-scheduled on analysis-only allocations have no data dependency among each other (they only communicate with their simulations), adjusting co-scheduling placements among them has no impact on the makespan. 
\subsection{Resource Allocation for \coalloc}
\label{sec:allocation}








\subsubsection{Rational Allocation}
We first show how to compute the (rational) numbers of nodes and cores assigned to each simulation and analysis in the workflow ensemble. In other words, we demonstrate how to solve \coalloc.
\rev{The idea is to allocate resources to co-scheduled applications such that differences among execution time of every co-scheduling allocations are minimized, thereby leading to minimal makespan. The intuition is that if one allocation has a smaller execution time than another allocation, we can take resources (thanks to rational number of nodes and cores) from the faster allocation to accelerate the slower one until all allocations finish approximately at the same time, hence improving the overall makespan.}

\pp{Analysis-only allocations}
\rev{Based on the above reasoning, we first compute resource allocation for analysis-only co-scheduling allocations (c.f.~\cref{tr:ns})}.
Let $P^{NC}$ denote the set of analyses that are not co-scheduled with their coupled simulation. 
Based on~\cref{tr:da}, 
we have to find a resource allocation for the co-scheduling mapping where analyses in $P^{NC}$ are co-scheduled on analysis-only co-scheduling allocations.
Let assume that the analysis in $P^{NC}$ are distributed among $L$ analysis-only co-scheduling allocations, in which each allocation is a no-intersecting subset of analyses denoted by $P_i^{NC}$. 
Given $X$ is a set of jobs, which are simulations or analyses, for the ease of presentation, we define the following notations:
%
\begin{itemize}
    \item $T(X) = \max_{x \in X} \; t(x)$ \rev{is the time to execute one iteration of $X$ concurrently;}
    \item $Q(X) = \sum_{x \in X} \; t_{x}(1)$ \rev{is the time to execute sequentially one iteration of $X$, i.e. executing on single core;} 
    \item $U(X) = \sum_{x \in X} \; c_{x} V_{x}$ \rev{is the cost of processing data remotely for $X$.}
\end{itemize}

\begin{theorem}
\label{tr:ns}
To minimize makespan, number of nodes and cores assigned to each analysis in analysis-only co-scheduling allocation $\langle P_i^{NC} \rangle$ are as follows:
\begin{align}
& n^{*}_{\langle P_i^{NC} \rangle} = \frac{B Q(P_i^{NC}) + U(P_i^{NC}) }{ B Q(S \cup A) + U(P^{NC})} N \label{eq:n_ns}\\
& c^{*}_{A_k} = \frac{ B Q(A_{k})}{B Q(P_i^{NC}) + U(P_i^{NC}) - C V_{A_{k}}} C  \; , \forall A_k \in P_i^{NC} \label{eq:c_ak} 
\end{align}
\end{theorem}
\begin{proof}
The proof is given in~\cref{apdx:ns}.
\end{proof}
Now the remaining part is to find $U(P_i^{NC})$.
To prevent from underutilized resources, for each analysis-only co-scheduling allocation containing $P_i^{NC}$, we expect   
    $\displaystyle \sum_{A_k \in P_i^{NC}} c_{A_k} = C$.
By substituting~\cref{eq:c_ak} to this equality, we have:
\begin{align}
    \label{eq:p_func}
    \sum_{A_k \in P_i^{NC}} \frac{ Q(A_{k})}{B Q(P_i^{NC}) + U(P_i^{NC}) - C V_{A_{k}}} = \frac{1}{B}
\end{align}
Because the left-hand side of~\cref{eq:p_func} is strictly decreasing when $U(P_i^{NC})$ is increased,~\cref{eq:p_func} has only a real root $U^{*}(P_i^{NC})$. 
In the experiments, we use SymPy, a Python package for symbolic computing to solve it numerically.
Substituting $U^{*}(P_i^{NC})$ to~\cref{eq:n_ns,eq:c_ak}, resource assignment for the analysis-only co-scheduling allocations is determined.

\pp{Simulation-based Allocations}
In this second step, we find a resource allocation for the remaining co-scheduling allocations, in which every simulation is co-scheduled with a subset of analyses it couples with. 
Specifically, we distribute remaining resources of $N - \naa$ nodes, where $\naa = \sum_{i=1}^{L} n^{*}_{\langle P_i^{NC} \rangle}$ (the value of $n^{*}_{\langle P_i^{NC} \rangle}$ is computed from~\cref{tr:ns}), for each simulation $S_i \in S$ and each analysis $A_j \in A \setminus P^{NC}$. Recall that $m(S_i)$ is the set of analyses that are co-scheduled with $S_i$.
\begin{theorem}
\label{tr:sb}
To minimize makespan, number of nodes and cores assigned to each simulation and analysis co-scheduled on co-scheduling allocation $\langle S_i \cup m(S_i) \rangle$ are as follows:
\begin{align}
n^{*}_{S_i} &= \frac{\displaystyle Q(S_{i} \cup m(S_i))}{Q(S \cup A \setminus P^{NC})}(N - \naa) \label{eq:n_si}\\
c^{*}_{S_i} &= \frac{Q(S_i)}{Q(S_i \cup m(S_i))} C \label{eq:c_si}\\
c^{*}_{A_j} &= \frac{Q(A_j)}{Q(S_i \cup m(S_i))} C \; , \forall A_j \in m(S_i) \label{eq:c_aj}
\end{align}
\end{theorem}
\begin{proof}
The proof is given in~\cref{apdx:sb}.
\end{proof}

\subsubsection{Integer Allocation}
\label{subsec:integer}
We have a solution to \coalloc with rational numbers of resources, however a more practically solution must use integer number of cores and nodes. 
Therefore, we relax the rational solution to an integer solution by applying a resource-preserving rounding heuristic. 
There are two levels of rounding needed to be handled: node-level rounding and core-level rounding. 
The node-level one rounds the number of nodes assigned to each co-scheduling allocation, while the core-level rounds the number of cores per node assigned to each application co-scheduled within the same co-scheduling allocation. 
The objective is to keep the makespan as minimized as possible after rounding.

These rounding problems are all declared as \emph{sum preserving rounding} as the sum of nodes or cores per node are the same after rounding. 
Formally, let assume we would like to round $x_i$ to an integer value $x_i^I$ such that $\sum x_i^I  = \sum x_i = s$.
For every $x_i$, a rounding way is to determine $i(x_i) \in \{0,1\}$ such that $x_i^I = x_i + i(x_i)$.
Hence, $\sum i(x_i) = s - \sum \lfloor x_i \rfloor$.
The idea is to pick $s - \sum \lfloor x_i \rfloor$ numbers among $x_i$ and round them up while round the others down, so that the sum of them are preserved.

The remaining task is to determine the number of nodes or cores of which jobs to round up based on the aforementioned objective of minimizing makespan, or minimizing the difference in new $T$, i.e. time to execute one iteration, after rounding (see~\cref{eq:makespan}).
From~\cref{eq:ts,eq:ta}, since assigned resources (i.e. number of nodes, number of cores per node) are inversely proportional to time to execution one iteration on one core, we present a rounding heuristic based on single-core execution time.
Specifically, at core-level, among applications $x$ in a co-scheduling allocation, we prioritize to round up $c_x$ whose $t_x(1)$ is larger.
At node-level, among co-scheduling allocations $X$, we prioritize to round up $n_{\langle X \rangle}$ whose $Q(X)$ is larger.
The algorithm following this rounding heuristic is described in~\cref{alg:rounding}.
%
\begin{algorithm}
    \caption{Rounding algorithm at node-level}\label{alg:rounding}
    \DontPrintSemicolon
    \Input{$N,\naa, n^{*}_{S_i}, Q(S_i \cup m^{*}(S_i))$}
    \Output{$n^{I}_{S_i}$}
    $k \gets N - \naa$\;
    $dict \gets \{\}$\;
    \For{$S_i \in S$}{
        $k \gets k - \lfloor n^{*}_{S_i} \rfloor$\;
        $dict \gets dict + \{S_i : Q(S_i \cup m^{*}(S_i))\}$\;
    }
    $sorted\_dict \gets sortByValue (dict)$\;
    $j \gets 1$\;
    \For{$S_i \in sorted\_dict$}{
        \eIf{$j > k$} {
            $n^{I}_{S_i} \gets \lfloor n^{*}_{S_i} \rfloor$\;
        }{
            $n^{I}_{S_i} \gets \lceil n^{*}_{S_i} \rceil$\;
        }
        $j \gets j + 1$\;
    }
    \Return{$n^{I}_{S_i}$}
\end{algorithm}
\section{Evaluation}
\label{sec:evaluation}
In this section, we present an evaluation of our proposed co-scheduling strategies developed to solve~\cosched and show the effectiveness of computing resource allocation using the approach developed to solve~\coalloc.

\pp{Simulator}
We implement a simulator based on WRENCH~\cite{wrench} and SimGrid~\cite{casanova2014simgrid} to accurately emulate (thanks to I/Os and network models validated by SimGrid) the execution of simulations and \insitu analyses in a workflow ensemble (see~\cref{fig:simulator}).
To emulate behavior of simulations and analyses, every simulation and analysis step is divided into two fine-grained stages: computational stage to perform certain simulation or analysis task, and I/O stage to publish/subscribe data to/from the simulation/analysis, respectively. 
The executions of these fine-grained stages tightly depend on each other following a validated \insitu execution model~\cite{do2021metric} to reflect coupling behavior between the simulation and \insitu analyses.
Task dependencies serve as an input for simulator workloads to simulate coupling behavior between simulations and \insitu analyses. 
Since a job in WRENCH is simulated on a single node, we emulate multi-node jobs by replicating stages of the job on multiple nodes and extrapolating task dependencies such that the execution order is satisfied by \insitu model.
The simulator is open-source and available online\footnote{\url{https://github.com/Analytics4MD/insitu-ensemble-simulator}}.
\begin{figure}[!ht]
	\centering
	\includegraphics[width=\linewidth]{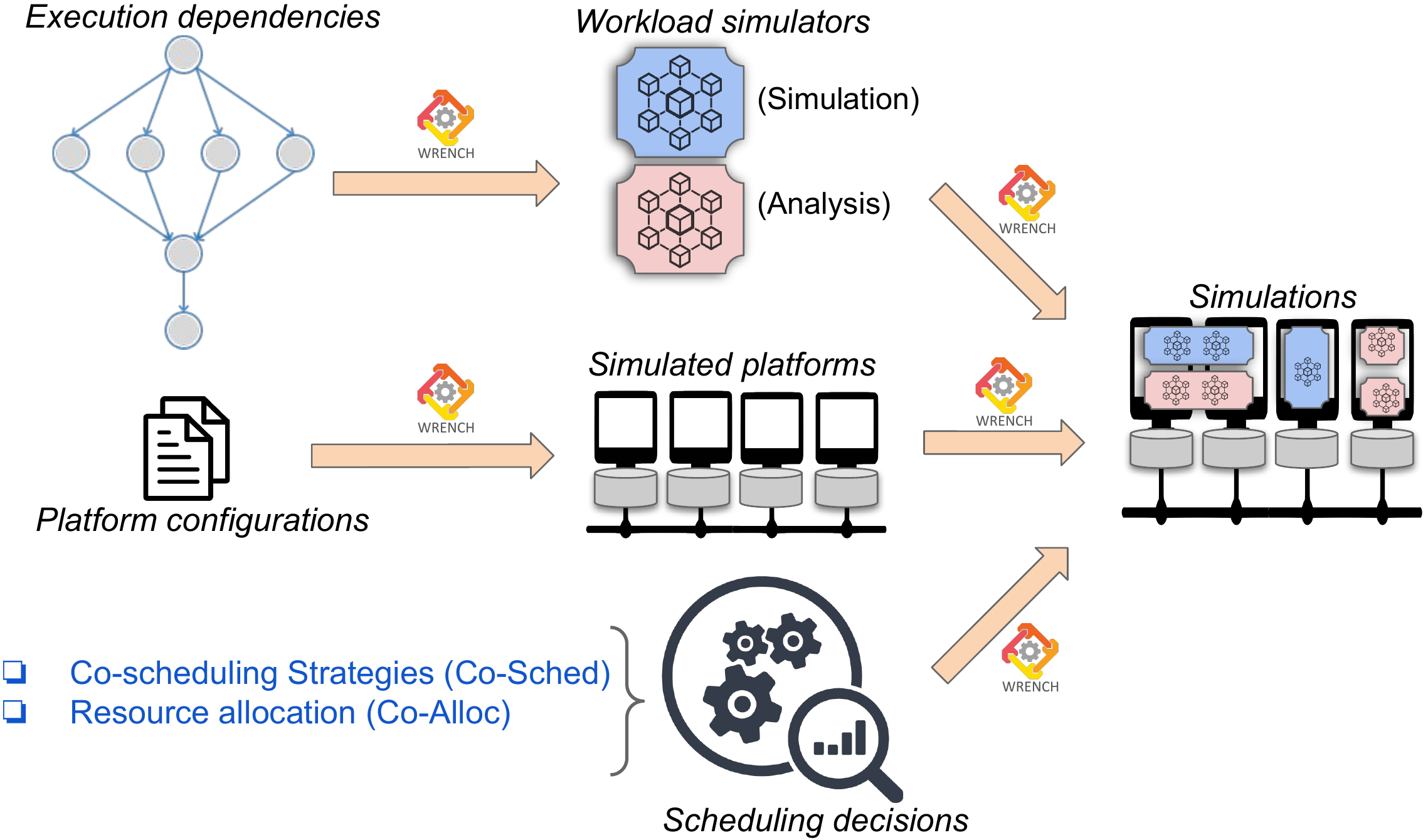}
	\caption{Simulator}
	\label{fig:simulator}
\end{figure}

\subsection{Experimental Setup}

\pp{Setup}
The simulator takes the following inputs: (1)~a workflow ensemble's configurations (i.e., number of simulations, number of coupled analyses per simulation, coupling relations), (2)~application profiles (i.e., simulations and analyses' sequential execution times), (3) resources constraints (i.e., number of nodes, number of cores per node, processor speed, network bandwidth), (3)~a co-scheduling mapping (i.e., the simulation each analysis is co-scheduled with), and (4)~a resource allocation (i.e., number of nodes and cores each application is assigned). 
We vary the workflow ensemble's configurations to evaluate the importance and the impact of different configurations on the efficiency of our proposed scheduling choices.
We fix the simulations' profiles, e.g. sequential execution time of the simulations, to study the impact of analyses on scheduling decisions. Sequential execution time of the analyses is generated randomly in a specific range which is relative to the sequential execution time of their coupled simulations. In this experiment, we generate it from $50\%$ to $150\%$ of the simulation's sequential execution time.
Even though our solution is applicable to any data size generated by the simulations, for the sake of simplicity, the amount of data processed by each analysis is kept at a fixed amount among all simulations, even though that amount is varied at different workflow ensemble's configurations.


\pp{Scenarios}
We compare a variety of scenarios corresponding to different co-scheduling mappings (see~\cref{table:scenarios}).
\ideal is the co-scheduling mapping where all analyses are co-scheduled with their respective simulation. 
\transit represents the scenario where all analyses are co-scheduled together on a dedicated co-scheduling allocation. 
We also compose hybrid scenarios where some analyses are co-scheduled with the simulation they couple with while the others are co-scheduled on the co-scheduling allocation without the presence of any simulation. 
For \increasing{x} (resp. \decreasing{x}), we pick the largest $x\%$ of the analyses, which is sorted in ascending (resp. descending) order of sequential execution time, to not co-schedule with the simulation they couple with (i.e. to co-schedule them on the dedicated co-scheduling allocation).
In this evaluation, we consider different percentages of number of analyses, $25\%$, $50\%$, $75\%$. 
For each of these co-scheduling mappings, we compute the resource allocations for the ensemble using \coalloc, which is described in~\cref{sec:allocation}.
%
%
\subsection{Bandwidth Calibration}
In order for the simulator to accurately reflect the execution platform, we have to calibrate the bandwidth per node used in our model to improve the precision of our solution using \coalloc. In our model, the bandwidth model is simple and optimistic as it does not account for concurrent accesses (i.e., the bandwidth available to each application is the maximum bandwidth $B$).
The idea is to compare the makespan given by our performance model (i.e., ~\cref{tr:ns}) to the makespan given by our simulation as WRENCH/SimGrid offers pretty complex and accurate communication models.
In~\cref{fig:bandwidth}, the theoretical makespan estimated by our model using the maximum bandwidth $B$ exhibits a large difference from the makespan resulted by the simulator (up to 30 times). The simulated bandwidth is smaller than $B$ during the execution as it is shared among concurrent I/O operations. 
Therefore, we explore different bandwidth values (see~\cref{table:bandwidth}) and re-estimate the makespan using our model based on the resource allocation at maximum bandwidth.

\begin{minipage}[b]{0.36\linewidth}
    \centering
    {
    \scriptsize
    \setlength{\tabcolsep}{3pt}
    \begin{tabular}[b]{cc}
        \toprule
        Model & Bandwidth \\
        \midrule 
        Baseline & $B$  \\
		$B'_{1}$ & $B / |P^{NC}|$  \\
		$B'_{2}$ & $B / \naa $  \\
		$B'_{3}$ & $B / (\naa \times |P^{NC}|)$ \\
        \bottomrule
    \end{tabular}
    }
    \captionof{table}{Bandwidth models to calibrate the simulator. \naa is number of nodes assigned to analysis in $P^{NC}$.}
    \label{table:bandwidth}
\end{minipage}
\hfill
\begin{minipage}[b]{0.56\linewidth}
    \vspace{5pt}
    \centering
    \includegraphics[width=\textwidth]{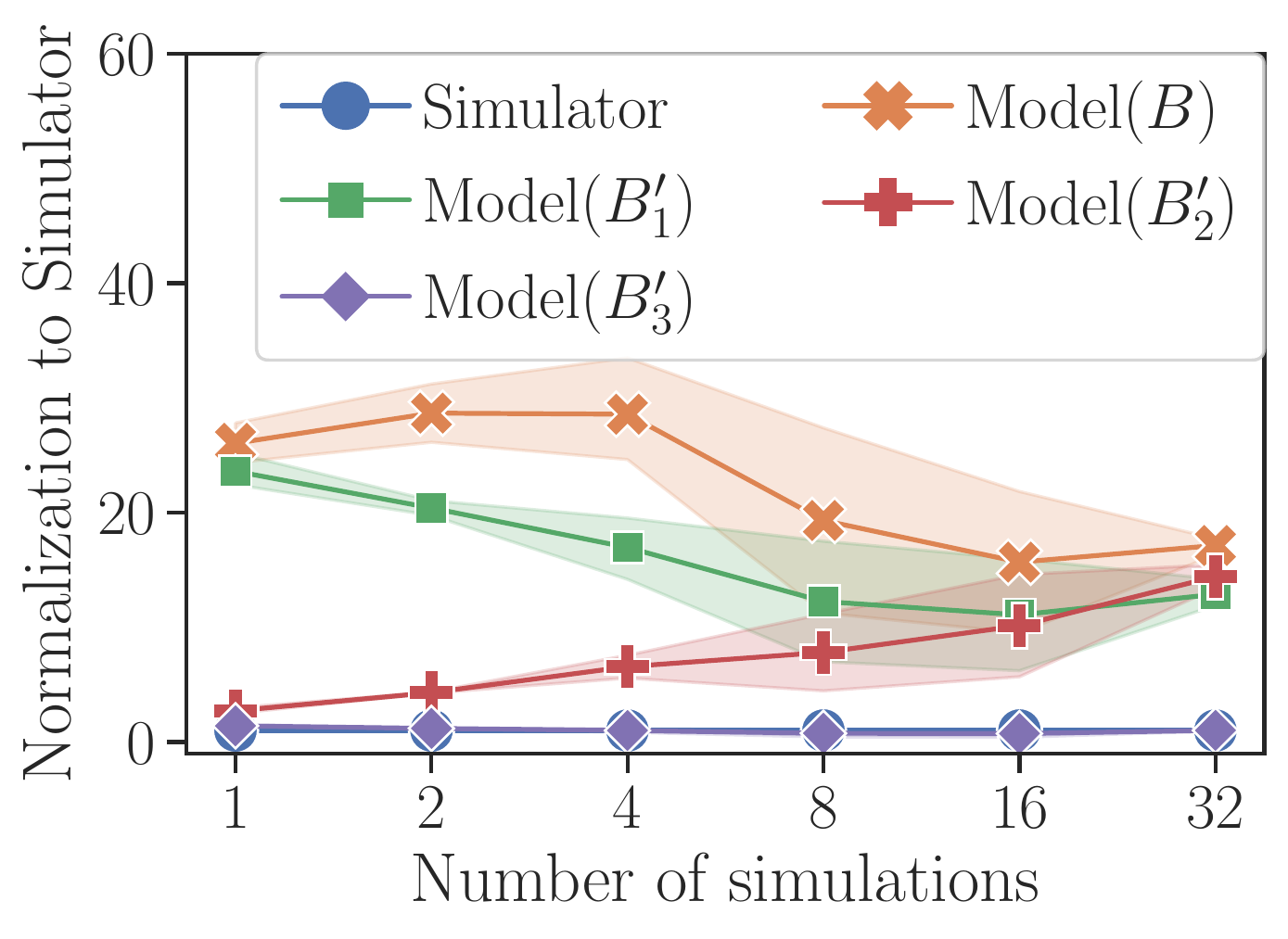}
    \captionof{figure}{Makespan estimated by the model for the \increasing{50} scenario with different values of bandwidth per node (see~\cref{table:bandwidth}). 
    }
    \label{fig:bandwidth}
\end{minipage}

~\cref{fig:bandwidth} confirms our hypothesis \rev{that bandwidth is shared among concurrent I/Os and the accuracy of \coalloc depends on choosing appropriate bandwidth used in the model}. $B'_{3}$ provides an approximation that is close to the makespan given by the simulator, we therefore choose $B'_{3}$ as the calibrated bandwidth for the  experiments hereafter.
%
%
\subsection{Results}
In this subsection, we explore the configuration space of the workflow ensemble. We aim to characterize the impact of different co-scheduling strategies by comparing the makespan of the simulator over various scenarios described in~\cref{table:scenarios}. 
To ensure the reliability of results, each value is averaged over 5 trials. 

\pp{Implications of \cosched}
%
%
%
\begin{table}[!t]
    \centering
\captionof{table}{Experimental scenarios, where $A(x\%)$ is the $x\%$ largest analyses' sequential time of $A$.}
{
    \centering
	\scriptsize
    \setlength{\tabcolsep}{4pt}
    \begin{tabular}{ccc}
    \toprule
    Scenarios & \shortstack{Analyses co-scheduled with\\their coupled simulation ($P^{C}$)} & \shortstack{Analyses not co-scheduled with\\their coupled simulation ($P^{NC}$)} \\
    \midrule 
	\ideal & A & \O \\
	\transit & \O & A \\
	\increasing{x} & $A \setminus A(x\%)$ & $A(x\%)$ \\
	\decreasing{x} & $A(100-x\%)$ & $A \setminus A(100-x\%)$ \\
    \bottomrule
    \end{tabular}
}
\label{table:scenarios}
\end{table}
\begin{figure}[!ht]
    \vspace{-15pt}
    \centering
    \begin{subfigure}[b]{\linewidth}
    	\centering
    	\includegraphics[width=\linewidth]{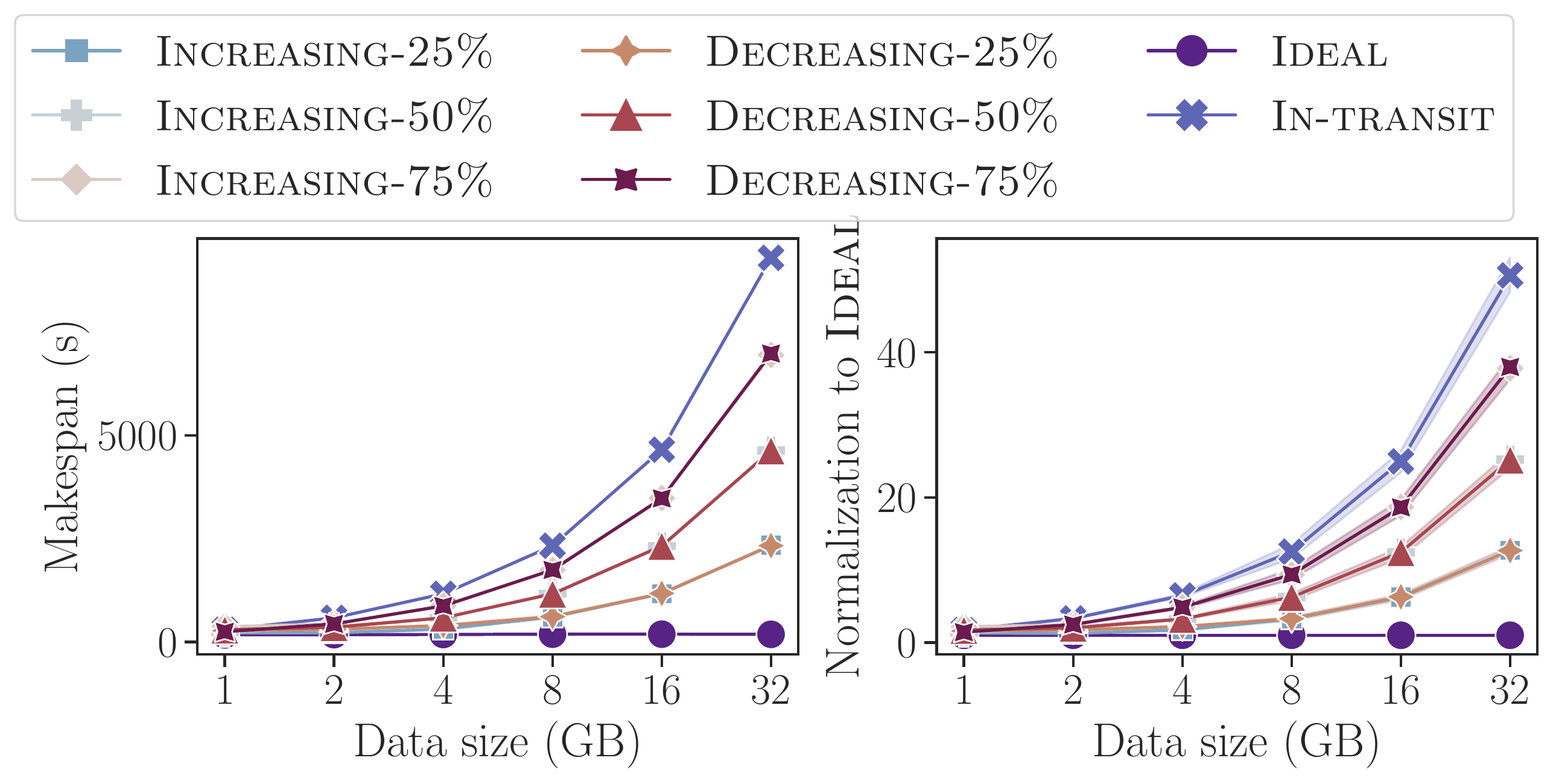}
    	\caption{Data size is varying. The workflow ensemble runs on $16$ nodes and has $4$ simulations, $4$ analyses per simulation.}
    	\label{fig:data}
    \end{subfigure}
    \begin{subfigure}[b]{\linewidth}
    	\centering
    	\includegraphics[width=\linewidth]{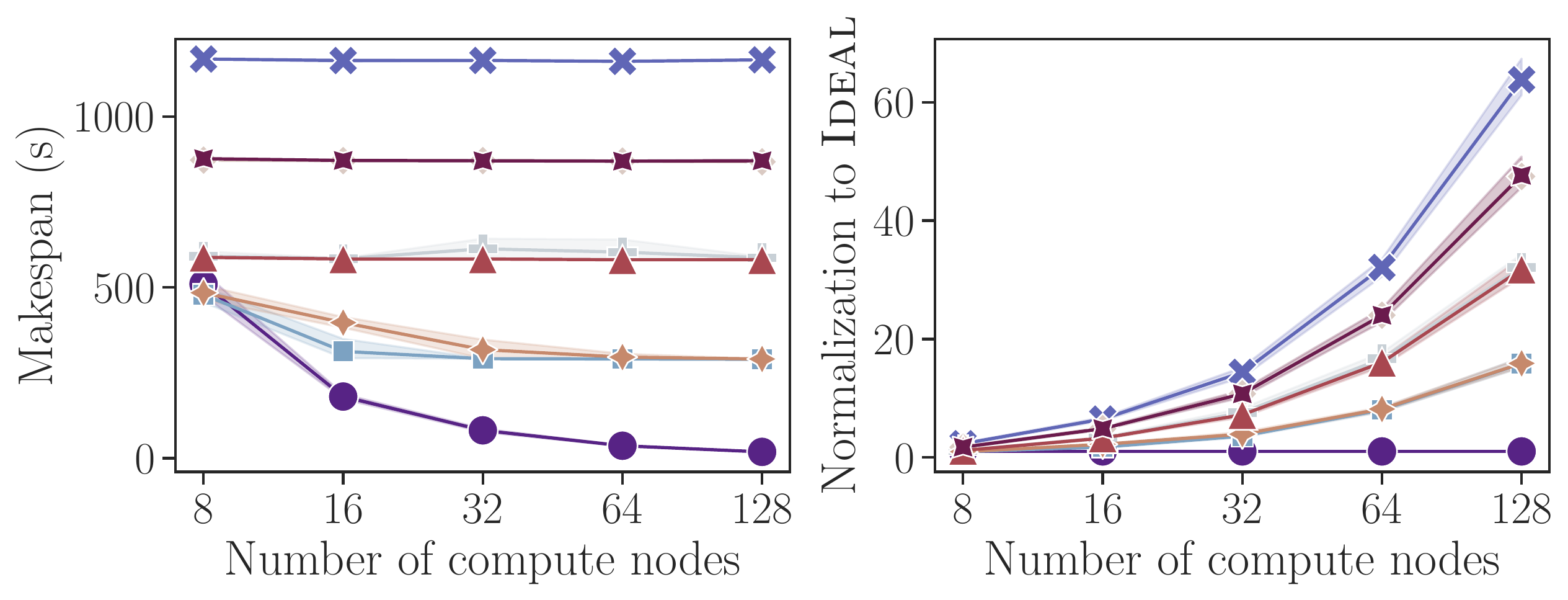}
    	\caption{Number of nodes is varying. The workflow ensemble has $4$ simulations, $4$ analyses per simulation. Each analysis processes $4GB$ of data each iteration}
    	\label{fig:node}
    \end{subfigure}
    \begin{subfigure}[b]{\linewidth}
    	\centering
    	\includegraphics[width=\linewidth]{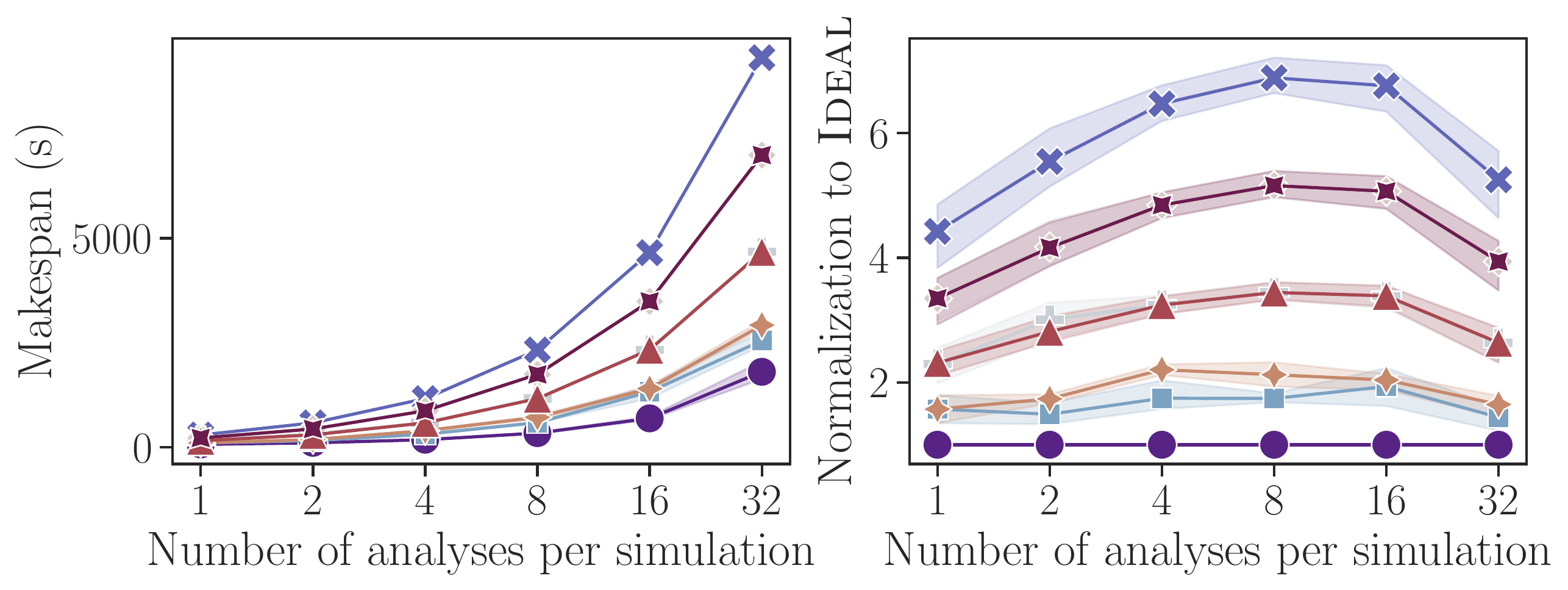}
    	\caption{Number of analyses per simulation is varying. The workflow ensemble runs on $16$ nodes and has $4$ simulations. Each analysis processes $4GB$ of data each iteration}
    	\label{fig:analysis}
    \end{subfigure}
    \begin{subfigure}[b]{\linewidth}
    	\centering
    	\includegraphics[width=\linewidth]{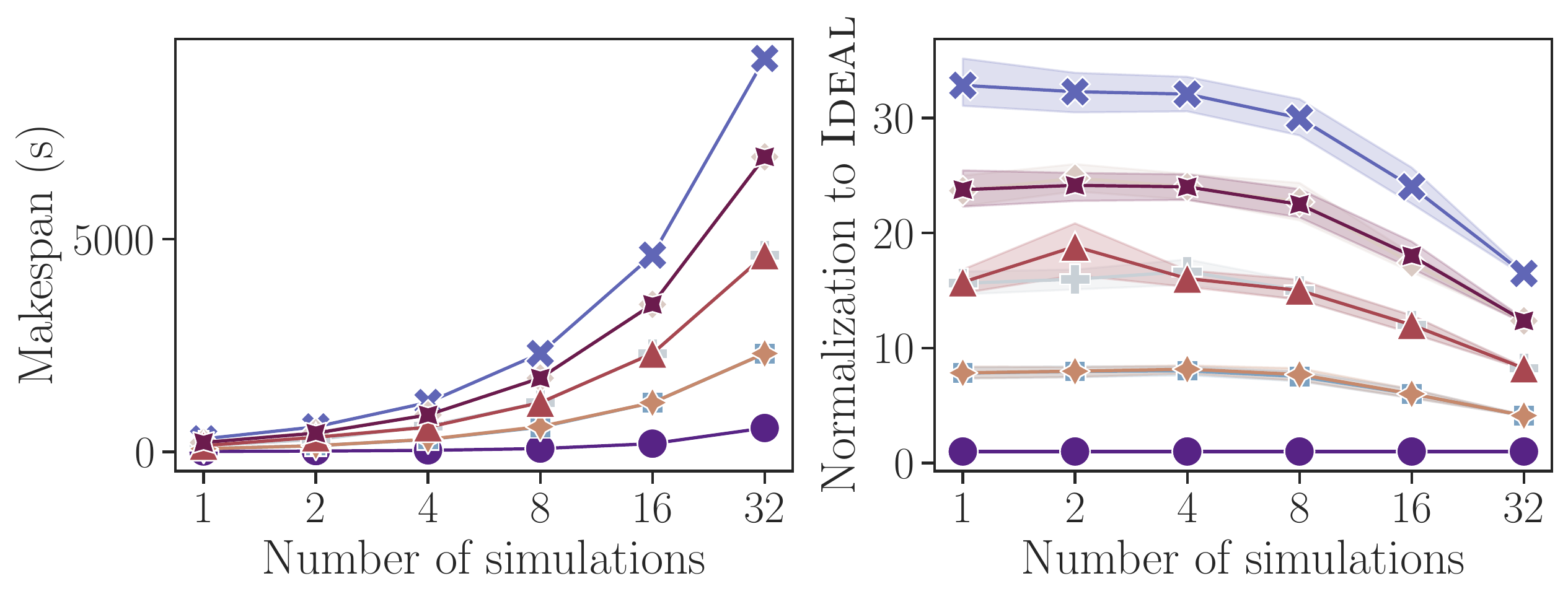}
    	\caption{Number of simulations is varying. The workflow ensemble runs on $64$ nodes and has $4$ analyses per simulation. Each analysis processes $4GB$ of data each iteration}
    	\label{fig:simulation}
    \end{subfigure}
    \caption{Makespan of the workflow ensemble when varying various configurations.}
    \label{fig:main}
    \vspace{-20pt}
\end{figure}
%
%
We vary one configuration of the workflow ensemble at a time while keeping the other configurations fixed. 
~\cref{fig:data} shows the makespan of the workflow ensemble running on $16$ nodes with $4$ simulations, $4$ analyses per simulation when varying the size of data processed by each analysis each iteration.
~\cref{fig:node} is the result with $4$ simulations, $4$ analyses per simulation and $4GB$ of data when varying the number of nodes.
In~\cref{fig:analysis}, the number of analyses per simulation is varied while keeping the number of simulation fixed at $4$ and $4GB$ of data, $4$ analyses per simulation.
In~\cref{fig:simulation}, the number of simulations is varied while keeping the number of analyses per simulation fixed at $4$, $4GB$ of data and $64$ nodes.

~\cref{fig:main} demonstrates the following order if the scenarios are sorted in ascending order by their makespan: \ideal, \textsc{x-25\%}, \textsc{x-50\%}, \textsc{x-75\%}, \transit, where \textsc{x} is either \textsc{Increasing} or \textsc{Decreasing} as those scenarios' makespan are approximately close to each other. Note that, for \ideal, $0\%$ of the analyses are not co-scheduled with their coupled simulation while that percentage is $100\%$ for \transit.
This observation aligns with our finding in~\cref{tr:gm} that not co-scheduling an analysis with its coupled simulation yields a higher makespan.
The more number of analyses not co-scheduled with their coupled simulation, the slower the makespan. 
Since \ideal outperforms the other scenarios, ideal co-scheduling mappings should be favored when the available resources can sustain.
In~\cref{fig:data,fig:analysis,fig:simulation}, the makespan scales linearly with the size of data processed by the analyses each iteration, the number of analyses per simulation and the number of simulations, respectively in which \ideal imposes the smallest escalation when these parameters are increased. 
In~\cref{fig:node}, only the makespan of the \ideal is decreased when the number of nodes is increased.
This is because the more nodes are utilized, the more communications are required to exchange the data, which requires the communication bandwidth is shared among them. The communication stages in each step are therefore slower even though the computational stages are faster due to more computing resources are assigned. Only the \ideal which incurs no remote communication can fully take advantage of the available resources. 

\pp{Efficiency of \coalloc}
%
\begin{figure}[!t]
	\centering
    \includegraphics[width=\linewidth]{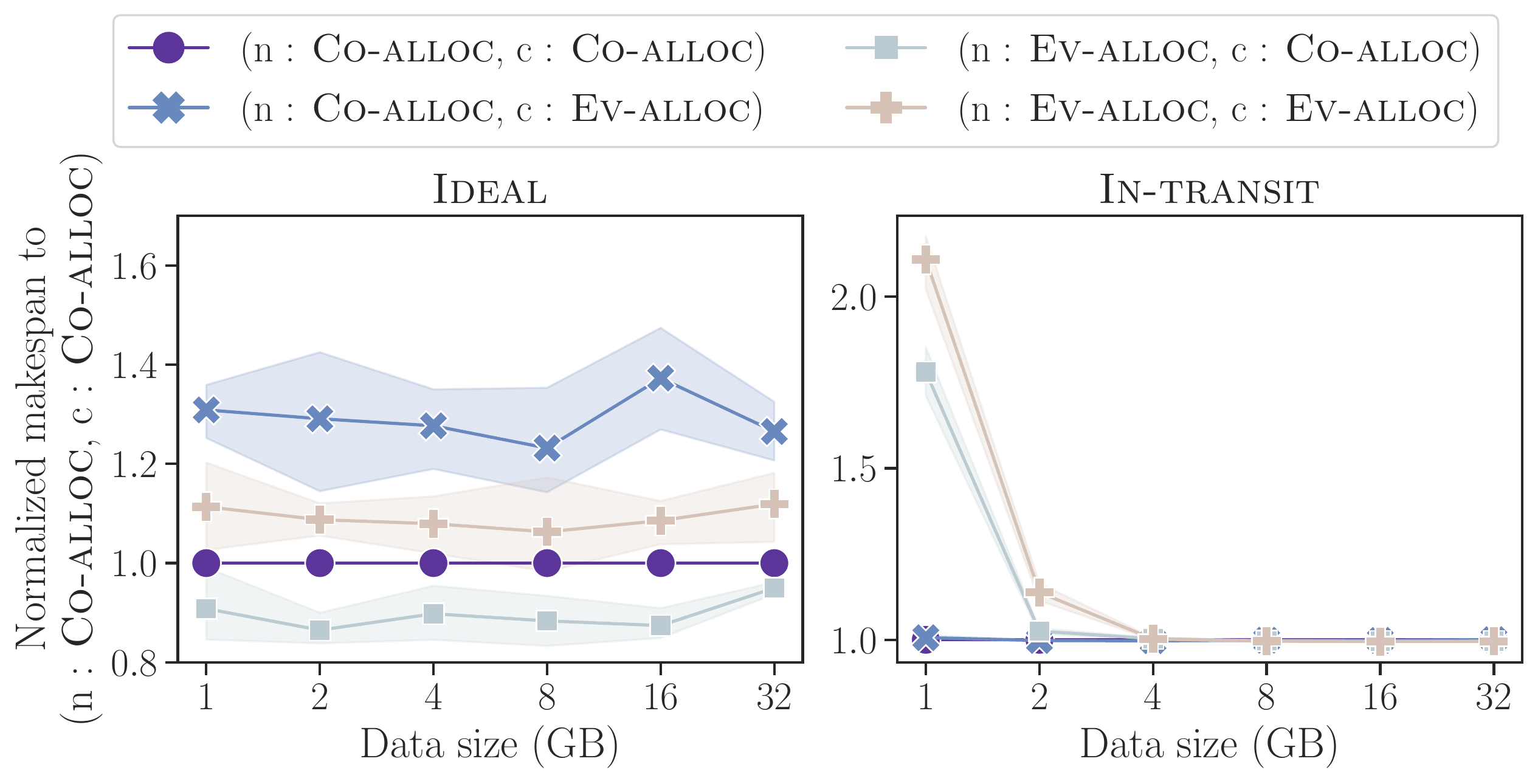}
	\caption{Comparison \rev{of normalized makespan} of \coalloc to \evalloc, where resources are evenly distributed. $(n:\coalloc, c:\evalloc)$ indicates \coalloc is applied at node-level while \evalloc is applied at core-level.}
	\label{fig:model}
\end{figure}
We compare the resource allocation computed by \coalloc with a naive approach of evenly distributing the resources, which is called as \evalloc. Specifically, at node-level, with \evalloc, the number of nodes is equally divided among co-scheduling allocations
while at core-level, for each co-scheduling allocation, the number of cores per node is equally divided among applications co-scheduled in such the allocation. 
Let $(n:x, c:y)$ is a manner to compute resource allocations, in which we apply $x$ method at node-level while using $y$ method at core-level. We study 4 possible combinations: $(n:\coalloc, c:\coalloc)$, $(n:\coalloc, c:\evalloc)$, $(n:\evalloc, c:\coalloc)$, and $(n:\evalloc, c:\evalloc)$.

Even though the experiment is conducted over all the scenarios in the configuration space, due to the lack of space, we selectively present a subset of scenarios where the allocations significantly differ between observing combinations. 
~\cref{fig:model} shows the makespan, which is normalized to $(n:\coalloc, c:\coalloc)$, in two representative extreme scenarios: \ideal and \transit when the amount of data processed each iteration is varied. 
In the \ideal scenario, $(n:\coalloc, c:\coalloc)$ surpasses the other combinations, except $(n:\evalloc, c:\coalloc)$.
This is because \ideal comprises no remote communication, thus execution time of each application's step is dominant by the portion of computational stages, which cause the performance is more sensitive to the changes at core-level than at node-level.
For \transit, using \evalloc at node-level results in a slower makespan (up to two times higher) than using \coalloc as off-node communications are primary in this scenario.
For larger data sizes, the I/O stages take the large proportion in each step in \transit, which makes the computing time negligible, thus there is no clear difference among discussed combinations.

\section{Conclusions}
\label{sec:conclusion}

In this paper, we have introduced an execution model of coupling behavior between \insitu jobs in a complex workflow ensemble. Based on the proposed model, we theoretically characterized computation and communication characteristics in each \insitu component. We further relied on this model to determine co-scheduling policies and resource profiles for each simulation and \insitu analysis such that the makespan of the workflow ensemble is minimized under given available computing resources.
By evaluating the scheduling solutions via simulation, we validate our proposed approach and confirm its correctness. We are also able to confirm the relevance of data locality as well as the need of well-management shared resources (e.g., bandwidth per node) in co-scheduling concurrent applications together.

For future work, we plan to extend the scheduling model to include contention and interference effects between \insitu jobs sharing the same underlying resources. 
Another promising research direction is to consider more complex use cases leveraging ensembles, e.g. adaptive sampling~\cite{vivek2020aebas} in which more coordination will be needed to achieve better performance as the simulations are periodically stopped and restarted.

\section*{Acknowledgments}
This work is funded by NSF contracts \#1741040, \#1741057, and \#1841758.


\appendix
Let $X$ is a set of jobs, which are simulations or analyses. We define the following:
%
\begin{itemize}
    \item $T(X) = \max_{x \in X} \; t(x)$ 
    \item $Q(X) = \sum_{x \in X} \; t_{x}(1)$ 
    \item $U(X) = \sum_{x \in X} \; c_{x} V_{x}$
\end{itemize}
\begin{remark}
\label{rm:tx}
To minimize $T(X) = \max_{x \in X} t(x)$ given that resources are finite, we allocate resources such that every $t(x)$, where $x \in X$, is equal to minimize difference among $t(x)$. 
\end{remark}
We use this remark to minimize execution time of concurrent applications. 
This remark is held when the number of resources (i.e. number of compute nodes, number of cores) are rational numbers.
We also note that all the following proofs are given under the context of rational allocation.
We discuss how to relax to integer allocations in~\cref{subsec:integer}.

Based on~\cref{rm:tx}, we compute minimal execution time in one iteration of a co-scheduling allocation, which is foundation for our following proofs.
\begin{lemma}
\label{tr:to}
Given a co-scheduling allocation $\langle S_i, m(S_i) \rangle$, and $P_{i}$ is the set of analyses that are co-scheduled with $S_i$, but not coupled with $S_i$.
The execution time for one iteration is minimized at:
\begin{equation}
    \label{eq:tca}
    T^{*}(S_i \cup m(S_i)) = \frac{\displaystyle B Q(S_i \cup m(S_i))  + U(P_i) }{ BC n_{S_i} } 
\end{equation}
\end{lemma}
\begin{proof}[Proof of~\cref{tr:to}]
For every analysis $A_j \in m(S_i) \setminus P_i$ (set of analyses co-scheduled with $S_i$ and coupled with $S_i$) and $A_k \in P_i$ (set of analyses co-scheduled with $S_i$, but not coupled with $S_i$), $T(S_i \cup m(S_i))$ is minimized at $T^{*}(S_i \cup m(S_i))$ when:
\begin{align}
    T^{*}(S_i \cup m(S_i)) = t(S_i) = t(A_j) = t(A_k)
\end{align}
Since $A_k \in P_i$ is not coupled with $S_i$, from~\cref{eq:ts,eq:ta}, we have:
\begin{align}
    \label{eq:to}
    &\frac{t_{S_i}(1)}{n_{S_i}c_{S_i}} = \frac{t_{A_{j}}(1)}{n_{A_{j}} c_{A_{j}}} = \frac{t_{A_k}(1)}{n_{A_k}c_{A_k}} + \frac{V_{A_k}}{Bn_{A_k}} \nonumber \\
    &= \frac{ B Q(S_i \cup m(S_i)) + U(P_i) }{\displaystyle B (n_{S_i}c_{S_i} + \sum_{A_j \in m(S_i)}n_{A_j}c_{A_j})} 
\end{align}
We note that 
    $n_{S_i} = n_{A_j}$, 
as they are co-scheduled on the same co-scheduling allocation. 
To avoid resources are underutilized, we expect the entire cores of every allocation are used, or:
\begin{equation}
    \label{eq:cop}
    c_{S_i} + \sum_{A_j \in m(S_i)}c_{A_j} = C
\end{equation}
Substituting~\cref{eq:cop} to~\cref{eq:to}, we obtain~\cref{eq:tca}, which is what we need to prove.
\end{proof}

\subsection{Proof of~\cref{tr:gm}}
\label{apdx:gm}
We compare the minimized makespan of two following co-scheduling mappings: (1) the original co-scheduling mapping $m$ 
and (2) the co-scheduling mapping $m'$ forming from $m$ by not co-scheduling an analysis $A_k$ with its coupled simulation. 
We are going to prove that the makespan of $m'$ is greater than one of $m$, which implies not co-scheduling an analysis 
with the its coupled simulation causes the makespan slower. We therefore, indirectly prove the theorem.  

Since $A_k$ is co-scheduled with its coupled simulation in the co-scheduling mapping $m$, from~\cref{rm:tx,eq:makespan}, $T(S \cup A)$ of the co-scheduling mapping $m$, is minimized when:
\begin{align}
    \label{eq:tak}
    T(S \cup A) = t(A_k) = \frac{t_{A_k}(1)}{n_{A_k}c_{A_k}} 
\end{align}
On the other hand, in the co-scheduling mapping $m'$, $A_k$ is not co-scheduled with its coupled simulation, then $T'(S \cup A)$ is minimized when: 
\begin{align}
    \label{eq:t'ak}
    T'(S \cup A) = t'(A_k) = \frac{t_{A_k}(1)}{n'_{A_k}c'_{A_k}} + \frac{V_{A_k}}{Bn'_{A_k}}
\end{align}
\pp{Proof by Contradiction}
Let assume $T(S \cup A) \geq T'(S \cup A)$, then $t(A_k) \geq t'(A_k)$. From~\cref{eq:tak,eq:t'ak}, we have:
\begin{align}
    \frac{t_{A_k}(1)}{n_{A_k}c_{A_k}} &\geq \frac{t_{A_k}(1)}{n'_{A_k}c'_{A_k}} + \frac{V_{A_k}}{Bn'_{A_k}} > \frac{t_{A_k}(1)}{n'_{A_k}c'_{A_k}} \nonumber \\
    &n_{A_k}c_{A_k} < n'_{A_k}c'_{A_k} \label{eq:r}
\end{align}
Let $r(x) = n_x c_x$ denote the amount of resources in terms of total cores assigned to $x$.
~\cref{eq:r} indicates more resources assigned to $A_k$ in $m'$ than $m$. 
However, as the total number of cores is limited to $N \times C$, in the co-scheduling $m'$, there must exist an application $I$ which is either a simulation in $S$ or an analysis in $A \setminus A_k$ such that it is assigned fewer resources than it is in the co-scheduling $m$. 
Hence, $T'(I) > T(I)$, which makes $T(S \cup A) < T'(S \cup A)$. This is contradicted to our assumption. 
Hence, from~\cref{eq:makespan}, the makespan of co-scheduling $m$ is smaller than the makespan of $m'$.

\subsection{Proof of~\cref{tr:da}}
\label{apdx:da}
We compare the minimized makespan of two following co-scheduling mappings: (1) the co-scheduling mapping $m^*$ where all analyses not co-scheduled with their coupled simulation, denoted by $P^{NC}$, are co-scheduled on analysis-only co-scheduling allocations and (2) 
the co-scheduling mapping $m'$ in which not all analyses in $P^{NC}$ are co-scheduled on analysis-only co-scheduling allocations. 
Hence, in this mapping $m'$, $P^{NC}$ is able to partition into two parts, $P'^{NC}$ is the set of analyses that are co-scheduled in analysis-only co-scheduling allocations, and $P'^{C}$ is the set of analyses that are co-scheduled on the co-scheduling allocations that has the simulations they do not couple with. 
Then $P'^{C} \cup P'^{NC} = P^{NC}$. 
We also denote the set of analyses co-scheduled with their coupled simulation in these co-scheduling mappings by $P^{C}$, thus $P^C \cup P^{NC} = A$. 
We are going to prove the makespan of $m'$ is greater than the makespan of $m^*$.

\pp{Co-scheduling mapping $m^*$}
For the mapping $m^{*}$, for every co-scheduling allocation $\langle S_i \cup m^*(S_i) \rangle$, from~\cref{tr:to}, $T(S_i \cup m^*(S_i))$ is minimized at:
\begin{align}
T^*(S_i \cup m^*(S_i)) = \frac{B Q(S_i \cup m^*(S_i))}{BCn^*_{S_i}}
\end{align}
Based on~\cref{rm:tx}, $T(S \cup P^C)$ is minimized when every $T^*(S_i \cup m^*(S_i))$ is equal, or
\begin{align}
\label{eq:t_spc}
T^*(S \cup P^C) = \frac{\displaystyle B \sum_{S_i \in S} Q(S_i \cup m^*(S_i))}{\displaystyle BC \sum_{S_i \in S} n^*_{S_i}} = \frac{BQ(S \cup P^C)}{\displaystyle BC \sum_{S_i \in S} n^*_{S_i}}
\end{align}
Following the same procedure for analysis-only co-scheduling allocations in $m^*$, from~\cref{tr:to}, $T(P^{NC})$ is minimized at:
\begin{align}
\label{eq:t_pnc}
T^*(P^{NC}) = \frac{ BQ(P^{NC}) + U^{*}(P^{NC}) }{ BC \naa^*}, 
\end{align}
where $\naa^*$ is total number of nodes assigned to analysis-only allocations.

\noindent
Now combining simulation-present co-scheduling allocations and analysis-only co-scheduling allocations, $T(S \cup A)$ is minimized, denoted by $t^*$, when $t^* = T^{*}(S \cup P^C) = T^{*}(P^{NC})$. From~\cref{eq:t_spc,eq:t_pnc}:
\begin{align}
    &t^* = \frac{ BQ(S \cup P^C) }{\displaystyle BC \sum_{S_i \in S}n^{*}_{S_i}} = \frac{ BQ(P^{NC}) + U^{*}(P^{NC}) }{ BC \naa^*} \nonumber\\
    &t^* = \frac{BQ(S \cup A) + U^{*}(P^{NC})}{NBC} \label{eq:t*} \\
    &U^*(P^{NC}) = \frac{BQ(S \cup P^C)}{\displaystyle \sum_{S_i \in S}n^{*}_{S_i}} N - BQ(S \cup A) \label{eq:u*}
\end{align}

\pp{Co-scheduling mapping $m'$}
Similarly, with the mapping $m'$, $T(S \cup A)$ is minimized, denoted by $t'$ when $t' = T'(S \cup P^C \cup P'^{C}) = T'(P'^{NC})$.
For the mapping $m'$, we use prime notations to differentiate from the mapping $m^*$.
\begin{align}
    &t' = \frac{\displaystyle B Q(S \cup P^C \cup P'^{C}) + U'(P'^{C}) }{\displaystyle BC \sum_{S_i \in S}n'_{S_i}} = \frac{\displaystyle BQ(P'^{NC}) + U'(P'^{NC}) }{ BC \naa' } \nonumber \\
    &t' = \frac{BQ(S \cup A) + U'(P^{NC})}{NBC} \label{eq:t'} \\
    &U'(P^{NC}) = \frac{\displaystyle BQ(S \cup P^C \cup P'^{C}) + U'(P'^{C}) }{\displaystyle \sum_{S_i \in S}n'_{S_i} } N - BQ(S \cup A) \label{eq:u'}
\end{align}

\pp{Compare $t'$ and $t^*$}
For the sake of simplicity, we denote:
\begin{itemize}
\item $n' = \sum_{S_i \in S}n'_{ S_i }$ and $n^* = \sum_{S_i \in S}n^{*}_{S_i}$ 
\item $U^* = U^*(P^{NC})$ and $U' = U'(P^{NC})$
\end{itemize} 
From~\cref{eq:t*,eq:t'}, we have:
\begin{equation}
    \label{eq:diff_t}
    t' - t^* = \frac{U'(P^{NC}) - U^*(P^{NC})}{NBC}
\end{equation}
Thus to compare $t^*$ and $t'$, we only need to compare $U^*(P^{NC})$ and $U'(P^{NC})$. 

\noindent
Moreover, from~\cref{eq:u*,eq:u'}: 
\begin{align}
    U' - U^* 
    = \frac{N}{n'}( \frac{BQ(S \cup P^C)}{n^*} (n^* -n')  +  U'(P'^{C}) + BQ(P'^{C}) ) \label{eq:u-u}
\end{align}
Again, from~\cref{tr:to}, we have:
\begin{equation}
\label{eq:t*bc}
t^* = \frac{BQ(S \cup P^C)}{BCn^*} \Rightarrow \frac{BQ(S \cup P^C)}{n^*} = t^*BC 
\end{equation}
Substituting~\cref{eq:t*bc} to~\cref{eq:u-u}:
\begin{equation}
U' - U^* = \frac{N}{n'}( t^*BC (n^* - n') + U'(P'^{C}) + BQ(P'^{C}) ) \label{eq:diff}
\end{equation}
Let us define the set of simulations that analyses in $P'^{C}$ co-scheduled with as $S'^{C}$. 
We note that $|S'^{C}| \leq |P'^{C}|$, then again from~\cref{tr:to}:
\begin{align}
    \label{eq:diff_p}
    & t' = T'(S'^C \cup P'^C) = \frac{ BQ(P'^{C}) + U'(P'^{C})}{\displaystyle BC \sum_{S_l \in S'^{C}} n_{S_l}} 
\end{align}
By substituting~\cref{eq:diff_t} to~\cref{eq:diff_p}:
\begin{align}
    \label{eq:helper}
    & BQ(P'^{C}) + U'(P'^{C}) = (t^*BC + \frac{U' - U^*}{N}) \sum_{S_l \in S'^{C}} n'_{S_l} 
\end{align}
By substituting~\cref{eq:helper} to~\cref{eq:diff}, we have:
\begin{align}
    \label{eq:diff_r}
    (U' - U^*)(\displaystyle n' - \sum_{S_l \in S'^{C}} n'_{S_l} ) = Nt^*BC(n^* + \sum_{S_l \in S'^{C}} n'_{S_l} - n')
\end{align}
Since the total number of compute nodes for co-scheduling $S$ and $P^C$ plus the number of nodes of co-scheduling allocations containing analyses in $P'^{C}$ is always greater than the total number of nodes for co-scheduling $S, P^C$ and $P'^{C}$, then 
\begin{align}
    n^*_{\langle S, P^C \rangle} + \sum_{S_l \in S'^{C}} n'_{\langle S_l, m'(S_l) \rangle} &> n'_{\langle S, P^C, P'^{C} \rangle} \nonumber \\
    n^* + \sum_{S_l \in S'^{C}} n'_{S_l} &> n' \label{eq:ine1} 
\end{align}
We also note that as $S'^C \subseteq S$ then: 
\begin{equation}
    \label{eq:ine2}
    n' = \displaystyle \sum_{S_i \in S}n'_{S_i} > \sum_{S_l \in S'^{C}} n'_{S_l}
\end{equation}
From~\cref{eq:diff_r,eq:ine1,eq:ine2}, we have $U' > U^*$, or $t' > t^*$.

\subsection{Proof of~\cref{tr:ns}}
\label{apdx:ns}
Let $\naa$ denote total number of nodes assigned to analysis-only co-scheduling allocations. Hence, $\naa = \sum_{i=1}^{L} n^*_{\langle P_i^{NC}\rangle}$, where $\cup_{i = 1}^{L} P_i^{NC} = P^{NC}$. We determine resource assignment to each analysis-only co-scheduling allocation at node- and core-level, respectively. 

\pp{Node-level Allocation}
To minimize makespan of entire ensemble (see~\cref{eq:makespan}), we need to minimize $T(S \cup A)$ by verifying that $T$ of analysis-only co-scheduling allocations is equal to $T$ of simulation-based co-scheduling allocations (according to~\cref{rm:tx}), or
    $T^{*}(S \cup P^{C}) = T^{*}(P^{NC})$.
By applying~\cref{tr:to}:
\begin{align}
    \frac{Q(S \cup P^C)}{C (N - \naa ) } 
    &= \frac{B Q(P^{NC}) + U(P^{NC})}{B C \naa} \nonumber \\
    \naa &= \frac{BQ(P^{NC}) + U(P^{NC})}{BQ(S \cup A) + U(P^{NC})} N \label{eq:tm_nc}
\end{align}
Among analysis-only co-scheduling allocations $\langle P_i^{NC} \rangle$, from~\cref{rm:tx}, to minimize $T(P^{NC})$, we verify $T^*(P^{NC}) = T^*(P_i^{NC}), \forall i \in \{ 1, \dots, L\}$. Similarly, applying~\cref{tr:to}:
\begin{align}
    \frac{B Q(P^{NC}) + U(P^{NC})}{B C \naa} 
    &= \frac{B Q(P_i^{NC}) + U(P_i^{NC})}{B C n^{*}_{\langle P_i^{NC} \rangle}} \label{eq:nc_nci}
\end{align}
Substituting~\cref{eq:tm_nc} to~\cref{eq:nc_nci}, the number of nodes assigned to each analysis-only co-scheduling allocation is computed as:
\begin{equation*}
    n^{*}_{\langle P_i^{NC} \rangle} = \frac{B Q(P_i^{NC}) + U(P_i^{NC}) }{ B Q(S \cup A) + U(P^{NC})} N 
\end{equation*}

\pp{Core-level Allocation}
Similarly, based on~\cref{rm:tx}, for each analysis-only co-scheduling allocation $\langle P_i^{NC} \rangle$, $T(P_i^{NC})$ is minimized when $T^*(P_i^{NC}) = T^{*}(A_k)$ for every $A_k \in P_i^{NC}$. From~\cref{eq:nc_nci,eq:ta}:
\begin{align}
\label{eq:h_nca}
\frac{B Q(P_i^{NC}) + U(P_i^{NC})}{B C n^*_{\langle P_i^{NC} \rangle}} &=  \frac{t_{A_{k}}(1)}{n^*_{A_k} c_{A_{k}} } + \frac{V_{A_{k}}}{B n^*_{A_k}} 
\end{align}
Note that $n^*_{\langle P_i^{NC} \rangle} = n^*_{A_k}$ as $A_k$ is co-scheduled on $\langle P_i^{NC} \rangle$. Hence, from~\cref{eq:h_nca}, the number of cores assigned to each analysis in analysis-only co-scheduling allocation expressed as:
\begin{equation*}
c^{*}_{A_k} = \frac{ B Q(A_{k})}{B Q(P_i^{NC}) + U(P_i^{NC}) - C V_{A_{k}}} C
\end{equation*}

\subsection{Proof of~\cref{tr:sb}}
\label{apdx:sb}
Similarly to the method used to prove~\cref{tr:ns}, we solve the number of nodes and cores per node assigned to each application in simulation-based co-scheduling allocations, respectively.

\pp{Node-level Allocation}
~\cref{rm:tx} shows that $T(S \cup P^C)$ is minimized when $T^{*}(S \cup P^C) = T^{*}(S_i \cup m(S_i))$ for every $S_i \in S$. From~\cref{tr:to}, the number of nodes assigned to each simulation-based co-scheduling allocation is computed $\langle S_i \cup m(S_i) \rangle$ as follows:
\begin{align}
    \frac{Q(S \cup P^C)}{C (N - \naa ) } &= \frac{Q(S_i \cup m(S_i))}{C n^{*}_{\langle S_i \cup m(S_i) \rangle}} \nonumber \\
    n^{*}_{\langle S_i \cup m(S_i) \rangle} &= \frac{\displaystyle Q(S_{i} \cup m(S_i))}{Q(S \cup P^C)}(N - \naa)
\end{align}

\pp{Core-level Allocation}
For each co-scheduling allocation, $T(S_i \cup m(S_i))$ is minimized when $T^{*}(S_i \cup m(S_i)) = T^{*}(S_i)$. Applying~\cref{tr:to}:
\begin{align}
    \frac{Q(S_i \cup m(S_i))}{C n^{*}_{\langle S_i \cup m(S_i) \rangle}} = \frac{Q(S_i)}{n^{*}_{S_i} c^{*}_{S_i}} \label{eq:sim_si}
\end{align}
As $S_i$ is co-scheduled in $\langle S_i \cup m(S_i) \rangle$, thus $n^{*}_{\langle S_i \cup m(S_i) \rangle} = n^{*}_{S_i}$. From~\cref{eq:sim_si}, the number of cores per node assigned to each simulation $S_i$ co-scheduled in $\langle S_i \cup m(S_i) \rangle$:
\begin{equation*}
c^{*}_{S_i} = \frac{Q(S_i)}{Q(S_i \cup m(S_i))} C
\end{equation*}
For every analysis $A_j \in m(S_i)$ is co-scheduled on co-scheduling allocation $\langle S_i \cup m(S_i) \rangle$, $T(S_i \cup m(S_i))$ is minimized when $T^{*}(S_i \cup m(S_i)) = T^{*}(A_j)$. From~\cref{tr:to}, the number of cores per node assigned to each analysis $A_j$ co-scheduled in $\langle S_i \cup m(S_i) \rangle$:
\begin{align}
&\frac{Q(S_i \cup m(S_i))}{C n^{*}_{\langle S_i \cup m(S_i) \rangle}} = \frac{Q(A_j)}{n^{*}_{A_j} c^{*}_{A_j}} \nonumber \\
&c^{*}_{A_j} = \frac{Q(A_j)}{Q(S_i \cup m(S_i))} C \nonumber
\end{align}

\bibliographystyle{IEEEtran}
\bibliography{short-reference}

\end{document}